\documentclass[12pt]{amsart} 
\usepackage{graphicx}
\usepackage{natbib}
\setcitestyle{aysep={}}
\usepackage{hyperref}
\usepackage{mathrsfs} %for math script font
\usepackage[nodisplayskipstretch]{setspace}
\usepackage[left=1.25in,top=1in,right=1.25in,bottom=1in,head=0.5in,foot=0.5in]{geometry}

\newcommand{\HH}{\mathcal{H}}
\newcommand{\CC}{\mathbb{C}}
\newcommand{\RR}{\mathbb{R}}
\newcommand{\Inn}[1]{\langle #1 \rangle}
\newcommand{\Norm}[1]{\left|\left\langle #1 \right\rangle \right|}

\newcommand{\paulix}{\bigl(\begin{smallmatrix} & 1 \\ 1 & \end{smallmatrix}\bigr)}
\newcommand{\pauliy}{\bigl(\begin{smallmatrix} & -i\\ i & \end{smallmatrix}\bigr)}
\newcommand{\pauliz}{\bigl(\begin{smallmatrix}1 & \\ & -1 
\end{smallmatrix}\bigr)}
\newcommand{\pzup}{\bigl(\begin{smallmatrix}1 \\ 0\end{smallmatrix}\bigr)}
\newcommand{\pzdown}{\bigl(\begin{smallmatrix} 0 \\ 1 \end{smallmatrix}\bigr)} 

\newtheorem{prop}{Proposition}
\newtheorem*{thm}{Theorem}
\theoremstyle{definition} 
\newtheorem{myth}{Myth} 

\makeatletter
\let\uppercasenonmath\@gobble% disables title uppercase
\makeatother

\title{Three Myths About Time Reversal in Quantum Theory}
\author[Bryan W. Roberts]{Bryan W. Roberts\\ \\Philosophy, Logic and Scientific Method\\Centre for Philosophy of Natural and Social Sciences\\London School of Economics and Political Science\\\href{mailto:b.w.roberts@lse.ac.uk}{b.w.roberts@lse.ac.uk} \\ \\ \today}
\date{\today. Accepted version, forthcoming in \emph{Philosophy of Science}.}
\thanks{\emph{Acknowledgements.} Thanks to Harvey Brown, Craig Callender, Tony Duncan, John Earman, Christoph Lehner, John D. Norton, Giovanni Valente, and David Wallace for helpful comments. Special thanks to David Malament for many helpful discussions during the development of these ideas. This work benefited from a National Science Foundation grant \# 1058902.} 

\begin{document}
\maketitle
\begin{abstract}
Many have suggested that the transformation standardly referred to as `time reversal' in quantum theory is not deserving of the name. I argue on the contrary that the standard definition is perfectly appropriate, and is indeed forced by basic considerations about the nature of time in the quantum formalism.
\end{abstract} 

\section{Introduction}

\subsection{Time reversal} Suppose we film a physical system in motion, and then play the film back in reverse. Will the resulting film display a motion that is possible, or impossible? This is a rough way of posing the question of time reversal invariance. If the reversed motion is always possible, then the system is time reversal invariant. Otherwise, it is not.

Unfortunately, the practice of reversing films is not a very rigorous way to understand the symmetries of time. Worse, it's not always clear how to interpret what's happening in a reversed film. The velocity of a massive body appears to move in the reverse direction, sure enough, but what happens to a wavefunction? What happens to an electron's spin? Such questions demand a more robust way to understand the meaning of time reversal in quantum theory. It's an important matter to settle, as the standard mathematical definition of time reversal plays a deep role in modern particle physics. One would like to have an account of the philosophical and mathematical underpinnings of this central concept. This paper gives one such account, which proceeds in three stages. We first show why time reversal is unitary or antiunitary, then that it is antiunitary, and finally uniquely derive the transformation rules.

\subsection{Controveries} The problem does not yet have an agreed-upon textbook answer. However, a prevalent response is that the transformation commonly referred to as `time reversal' in quantum theory isn't really deserving of the name. Eugene Wigner, in the first textbook presentation of time reversal in quantum mechanics, remarked that ```reversal of the direction of motion' is perhaps a more felicitous, though longer, expression than {`time inversion'}'' \cite[p.325]{wigner1931}. Later textbooks followed suit, with \citet[p.266]{sakurai1994} writing: ``This is a difficult topic for the novice, partly because the term \emph{time reversal} is a misnomer; it reminds us of science fiction. Actually what we do in this section can be more appropriately characterized by the term \emph{reversal of motion}.'' And in  \citet[p.377]{ballentine1998} we find, ``the term `time reversal' is misleading, and the operation... would be more accurately described as \emph{motion reversal}.''

Some philosophers of physics adopted this perspective and ran with it. \citet{callender2000time} suggests that we refer to the standard definition as `Wigner reversal', leaving the phrase `time reversal' to refer to the mere reversal of a time ordering of events $t\mapsto -t$. This leads him to the radical conclusion\footnote{One precise way to derive this conclusion is as follows. Let $\psi(t)$ be any solution to the Schr\"odinger equation $i\tfrac{d}{dt}\psi(t) = H\psi(t)$, where $H$ is a fixed self-adjoint and densely-defined operator. Suppose that if $\psi(t)$ is a solution, then so is $\psi(-t)$, in that $i\tfrac{d}{dt}\psi(-t) = H\psi(-t)$. We have by substitution $t\mapsto-t$ that $-i\tfrac{d}{dt}\psi(-t) = H\psi(-t)$, and so adding these two equations we get $0 = 2H\psi(-t)$ for all $\psi(t)$. That is only possible if $H$ is the zero operator. So, if a quantum system is non-trivial (i.e. $H\neq0$), then it is not invariant under $\psi(t)\mapsto\psi(-t)$. In contrast, most familiar quantum systems \emph{are} invariant under the standard time reversal transformation, $\psi(t)\mapsto T\psi(-t)$.} that, not just in weak interactions like neutral kaon decay, but in ordinary non-relativistic Schr\"odinger interactions, the ``evolution is not TRI [time reversal invariant], contrary to received wisdom, so time in a (nonrelativistic) quantum world is handed'' \citep[p.268]{callender2000time}.

%``\emph{non-relativistic} quantum mechanics \emph{already} tells us that Nature cares about time reversal'' in that it indicates ``a fundamentally irreversible world'' \citep[p.249]{callender2000time}.

Albert adopts a similar perspective, writing, ``the books identify precisely that transformation as the transformation of `time-reversal.' ... The thing is that this identification is \emph{wrong}. ... [Time reversal] can involve nothing whatsoever other than reversing the \emph{velocities of the particles}'' \citep[pgs. 20-21]{albert2000}. This implies that time reversal cannot conjugate the wavefunction, as is standardly assumed, which leads Albert to declare, ``the dynamical laws that govern the evolutions of quantum states in time cannot possibly be invariant under \emph{time-reversal}'' (p.132). A detailed critical discussion of Albert's general perspective has been given by \citet{earman2002}.

Both Callender and Albert argue that there is something unnatural about supposing time reversal does more than reverse the order of states in a trajectory. The standard expression of time reversal maps a trajectory $\psi(t)$ to $T\psi(-t)$, reversing the order of a trajectory $t\mapsto-t$, but also transforming instantaneous states by the operator $T$. Both Callender and Albert propose that time reversal is more appropriately described by mere `order reversal' $\psi(t)\mapsto\psi(-t)$. Callender explains the view as follows.
\begin{quote}
  David Albert... argues --- rightly in my opinion --- that the traditional definition of [time reversal invariance], which I have just given, is in fact gibberish. It does not make sense to \emph{time-reverse} a truly \emph{instantaneous} state of a system. \citep[p.254]{callender2000time}
\end{quote}
Some quantities, such as a velocity $dx/dt$, may still be reversed. However, the view is that these are not truly instantaneous quantities, but depend in an essential way on the directed development of some quantity in time. A quantity that is truly only defined by an instant cannot be sensibly reversed by the time reversal transformation. One might refer to the underlying concern as the `pancake objection': if the evolution of the world were like a growing stack of pancakes, why should time reversal involve anything other than reversing the order of pancakes in the stack?

Here is one reason: properties at an instant often depend essentially on temporal direction, even though this may not be as apparent as in the case of velocity. Consider the case of a soldier running towards a vicious monster. In a given instant, someone might call such a soldier `brave' (or at least `stupid'). The time-reversed soldier, running away from the vicious monster, would more accurately be described as `cowardly' at an instant. The situation in fundamental physics is analogous: properties like momentum, magnetic force, angular momentum, and spin all depend in an essential way on temporal direction for their definition. The problem with the pancake objection is that it ignores such properties: time reversal requires taking each individual pancake and `turning it around', as it were, in addition to reversing the order.

A supporter of Callender and Albert could of course deny that there are good reasons to think that bravery, momentum or spin are intrinsically tied to the direction of time. Callender refers to many such suggestions as ``misguided attempts,'' arguing on the contrary that from a definition of momentum such as $P=i\hbar\tfrac{d}{dx}$ in the Schr\"odinger representation, the lack of appearance of a `little $t$' indicates that it ``does not logically follow, as it does in classical mechanics, that the momentum... must change sign when $t\mapsto-t$. Nor does it logically follow from $t\mapsto-t$ that one must change $\psi\mapsto\psi^*$'' \citep[p.263]{callender2000time}.

I am not convinced. There is a natural perspective on the nature of time according to which quantities like momentum and spin really do change sign when time-reversed, or so I will argue. This may not be obvious from their expression in a given formalism. But, as \citet{malament2004} has shown, some quantities (like the magnetic field) may depend on temporal direction even when there is not an obvious `little $t$' in the standard formalism\footnote{Malament illustrates a natural sense in which the magnetic field $B^a$ is defined by the Maxwell-Faraday tensor $F_{ab}$, which is in turn defined with respect to a temporal orientation $\tau^a$. So, since time reversal maps $\tau^a\mapsto-\tau^a$, it follows that it $F_{ab}\mapsto-F_{ab}$ and $B^a\mapsto-B^a$ \citep[\S 4-6]{malament2004}.}. I claim that the situation is similar in quantum theory, and that consequently, it is no less natural to reverse momentum or spin under time reversal than it is to reverse velocity. 

Let me set aside arguments from monsters and other gratuitous metaphors for the remainder of this paper. I only give them to provide some physical intuition for those who find it helpful. My aim here is more general. In what follows I will set out and motivate a few precise elementary considerations about the nature of time, and then show how they lead inevitably to the standard definition of time reversal in quantum mechanics, complete with the standard transformation rules on instantaneous states. Along the way I will seek to dissolve three myths about time reversal in quantum theory, which may be responsible for some of the controveries above.  

\subsection{Three Myths} The skeptical perspective, that the standard definition of time reversal is not deserving of the name, arises naturally out of three myths about time reversal in quantum theory. In particular, these myths suggest that the justification for the standard time reversal operator amounts to little more than a convention. If that were true, then one could freely propose an alternative definition as Callender and Albert have done, without loss. I will argue that there is something lost. The standard definition of time reversal cannot be denied while maintaining a plausible perspective on the nature of time. It is more than a convention, in the sense that the following three myths can be dissolved.

\begin{myth}
\emph{The preservation of transition probabilities ($\Norm{T\psi,T\phi}=\Norm{\psi,\phi}$) is a conventional feature of time reversal, with no further justification.} Many presentations presume this is just a conventional property of `symmetry operators.' A common myth is that there is no good answer to the question of \emph{why} such operators preserve transition probabilities. I will point out one good reason.
\end{myth}

\begin{myth}
  \emph{The antiunitary (or `conjugating aspect') of time reversal is a convention, unjustified, or else presumes certain transformation rules for `position' and `momentum.'} When it is not posited by convention, one can show that antiunitarity follows from the presmption that time reversal preserves position ($Q \mapsto Q$) and reverses momentum ($P \mapsto -P$), as we shall see. This argument has unfortunate limitations. I will propose an improved derivation.
\end{myth}

\begin{myth}
\emph{The way that time reversal transforms observables is a convention, unjustified, or requires comparison to classical mechanics.} When asked to justify the transformations $Q \mapsto Q$ and $P \mapsto -P$, or the claim that $T^2=-1$ for odd-fermion systems, authors often appeal to the myth that this is either a convention, or needed in order to match the classical analogues in Hamiltonian mechanics. I will argue neither is the case, and suggest a new way to view their derivation.
\end{myth}

Callender and Albert have fostered the second myths in demanding that time reversal invert the order of instantaneous states without any kind of conjugation; they have fostered the third in arguing that it doesn't necessarily transform momentum and spin\footnote{Callender argues that momentum reverses sign in quantum theory only because of a classical correspondence rule. I discuss this argument in detail in Section \ref{subs:pos-mom} below.}. However, these perspectives aside, I hope that the dissolution of these myths and the account of time reversal that I propose may be of independent interest. In place of the myths I will give one systematic way to motivate the meaning of time reversal in quantum theory, and argue that it is justifiably associated with the name. As in the case of Malament's perspective, skeptics may still wish to adopt alternatives to the standard use of the phrase `time reversal'. Fine: one is free to define terms how one chooses. But as with Malament's perspective on electromagnetism, this paper will aim to show just how much one is giving up by denying the standard definitions. The account builds up the meaning of time reversal in three stages, dissolving each of the three myths in turn along the way.

\section{First Stage: Time reversal is unitary or antiunitary}\label{sec:firststage}

\subsection{Wigner's theorem} Wigner's theorem is one of the central results of modern quantum theory, first presented by \citet{wigner1931}. The theorem is often glossed as showing that any transformation $A:\HH\rightarrow\HH$ on a separable Hilbert space that deserves to be called a `symmetry' must be unitary or antiunitary. The statement is more accurately put in terms of \emph{rays}, or equivalence classes of vectors related by a phase factor, $\Psi := \{ e^{i\theta}\psi \;|\; \psi\in\HH \text{ and } \theta\in\RR \}$. Since each vector $\psi$ in a ray gives the same expectation values, it is often said that rays are what best represent `physical' quantum states. There is an inner product on rays defined by the normed Hilbert space inner product $\Inn{\Psi,\Phi} := |\Inn{\psi,\phi}|$, where $\psi\in\Psi$ and $\phi\in\Phi$; this product is independent of which vectors in the rays are chosen. What Wigner presumed is that every symmetry, including time reversal, can be represented by a transformation $\mathbf{S}$ on rays that preserves the inner product, $\Inn{\mathbf{S}\Psi,\mathbf{S}\Psi} = \Inn{\Psi,\Phi}$. From this he argued for Wigner's Theorem, that every such transformation can be uniquely (up to a constant) implemented by either a unitary operator or an antiunitary operator.

Time reversal, as we shall see in the next section, falls into the latter `antiunitary' category. But before we get that far: why do we expect time reversal to preserve inner products between rays? Or, in terms of the underlying Hilbert space vectors, why should time reversal preserve transition probabilities? Of course, Wigner is free to define words however he likes. But one would like to have a more serious reason.

\subsection{Uhlhorn's theorem} Here is a general way to answer this question that I think is not very well-known. To begin, consider two rays that are orthogonal, $\Inn{\Psi,\Phi}=0$. In physical terms, this means that the two corresponding states are mutually exclusive: if one of them is prepared, then the probability of measuring the other is zero, in every experiment. To have a simple model in mind: take $\Psi$ and $\Phi$ to represent $z$-spin-up and $z$-spin-down eigenstates, which are orthogonal in this sense.

Suppose we interpret a `symmetry transformation' to be one that preserves orthogonality. For example, since $z$-spin-up and $z$-spin-down are mutually exclusive outcomes in an experiment, we suppose that this will remain the case when the entire experimental setup is `symmetry-transformed', by say a rigid rotation or by a translation in space. And vice versa: if two symmetry-transformed states are mutually exclusive, then we assume the original states must have also been mutually exclusive. In the particular case of quantum mechanics, we thus posit the following natural property of symmetry transformations: two rays are orthogonal if and only if the symmetry transformed states are too. \citet{uhlhorn1963} discovered that, surprisingly, this requirement is enough to establish that symmetries are unitary or antiunitary\footnote{A concise proof is given by \cite[Theorem 4.29]{varadarajan-geom}; I thank David Malament for pointing this out to me. Uhlhorn's theorem was considerably generalised by \citet{molnar2000generalwigner}; see \citet{chevalier2007a} for an overview.} (when the dimension of the Hilbert space is greater than 2).

\begin{thm}[Uhlhorn]\label{thm:Uhlhorn}
Let $\mathbf{T}$ be any bijection on the rays of a separable Hilbert space $\HH$ with $\dim\HH>2$. Suppose that $\Inn{\Psi,\Phi}=0$ if and only if $\Inn{\mathbf{T}\Psi,\mathbf{T}\Phi}=0$. Then,
\begin{equation*}
	\Inn{\mathbf{T}\Psi,\mathbf{T}\Phi}=\Inn{\Psi,\Phi}.
\end{equation*}
Moreover, there exists a unique (up to a constant) $T:\HH\rightarrow \HH$ that implements $\mathbf{T}$ on $\HH$ in that $\psi \in \Psi$ iff $T\psi \in \mathbf{T}\Psi$, where $T$ is either unitary or antiunitary and satisfies $\Norm{T\psi,T\phi} = \Norm{\psi,\phi}$ for all $\psi, \phi \in \HH$.
\end{thm}
In other words, as long as a transformation preserves whether or not two states are mutually exclusive, it must either be unitary or antiunitary.

\subsection{Time reversal} The interpretation of Uhlhorn's theorem is perspicuous in the special case of time reversal, where it immediately dissolves our first myth. Suppose some transformation can be interpreted as involving `reversal of the direction of time'. That is, I wish to speak not just of `motion reversal' as some textbooks prefer to say, but `time reversal', whatever that should mean. Whatever else one might say about time reversal, let us at least suppose that two mutually exclusive states remain so under the time reversal transformation, in that the states $\Psi$ and $\Phi$ are orthogonal if and only if $T\Psi$ and $T\Phi$ are too. Why believe this, when nobody has ever physically `reversed time'? The reason is that \emph{whether two states are mutually exclusive has nothing to do with the direction of time}. Orthogonality is a statement about what is possible in an experimental outcome, independently of their time-development. Accepting this does not require any kind of lofty metaphysical indulgence. Orthogonality is simply not a time-dependent concept.

This is all that we need. We can immediately infer that time reversal preserves transition probabilities, and is implemented by a unitary or antiunitary operator. That is the power of Uhlhorn's theorem. Contrary to the first myth, there is indeed a reason to accept that time reversal preserves transition probabilities and thus is unitary or antiunitary. It emerges directly out of a reasonable constraint on what it means to reverse time, together with the mathematical structure of quantum theory.

\section{Second Stage: Why $T$ is Antiunitary}

\subsection{Antiunitarity} We have argued that time reversal must be unitary or antiunitary. But the standard definition further demands that it is antiunitary in particular. An antiunitary operator is a bijection $T:\HH\rightarrow\HH$ that satisfies,
\begin{enumerate}
  \item (adjoint inverse) $T^*T=TT^*=I$, and
  \item (antilinearity) $T(a\psi + b\phi) = a^*T\psi + a^*T\phi$.
\end{enumerate}
It is sometimes useful to note that these conditions are together equivalent to,
\begin{enumerate}
  \item[(3)] $\Inn{T\psi,T\phi}=\Inn{\psi,\phi}^*$.
\end{enumerate}
Properties (2) and (3) underlie claims that time reversal `involves conjugation'. They are also slippery properties that often throw beginners (and many experts) for a loop, since they require many of the familiar properties of linear operators to be subtly adjusted.

When is a transformation antiunitary, as opposed to unitary? It is not the `discreteness' of the transformation, since the parity transformation is discrete and unitary. It is rather a property that holds of all `time-reversing' transformations, including $T$, $PT$, $CPT$, and indeed any $UT$ where $U$ is a unitary operator. Once one has accepted that the time reversal operator $T$ is an antiunitary bijection, it follows that the transformations of the form $UT$ are exactly the antiunitary ones: if $T$ is antiunitary, then so is $UT$ when $U$ is unitary; and conversely, if $A$ is any antiunitary operator, then there exists a unitary $U$ such that $A = UT$, as one can easily check\footnote{The former follows immediately from the definition; the latter follows by setting $U := AT^{-1}$ and checking that $U$ is unitary.}.

Some of the mystery about antiunitary operators can be dissolved by noting that there is a similar property in classical Hamiltonian mechanics. In local coordinates $(q,p)$, interpreted as position and momentum, the instantaneous effect of time reversal is normally taken to preserve position and to reverse momentum $(q,p)\mapsto(q,-p)$. But it is easy to check that it is not a canonical transformation. The mathematical reason for this is that time reversal does not preserve the symplectic form $\omega = dq\wedge dp$, which is the geometric structure underpinning Hamilton's equations. Instead, the symplectic form reverses sign under time reversal. For this reason, time reversal in classical Hamiltonian mechanics is more correctly identified `anticanonical' or `antisymplectic', which is directly analogous to antiunitarity in quantum mechanics.

\citet{earman2002} has offered some `physical' motivation for an antiunitary time reversal operator in quantum mechanics:
\begin{quote}
[T]he state $\psi(x,0)$ at $t=0$ not only determines the probability distribution for finding the particle in some region of space at $t=0$ but it also determines whether at $t=0$ the wave packet is moving, say, in the $+x$ direction or in the $-x$ direction. ...
So instead of making armchair philosophical pronouncements about how the state cannot transform, one should instead be asking: How can the information about the direction of motion of the wave packet be encoded in $\psi(x,0)$? Well (when you think about it) the information has to reside in the phase relations of the components of the superposition that make up the wave packet. And from this it follows that the time reversal operation must change the phase relations.
\end{quote}
In short, phase angles in quantum theory contain information that is temporally directed. As a consequence, one cannot reverse time without reversing those phase angles. This is precisely what an antiunitary operator does, since $Te^{i\theta}\psi = e^{-i\theta}T\psi$.

I find Earman's motivation compelling. However, one would like to have a more general and systematic derivation of antiunitarity. I will consider two such derivations below. The first is a common textbook argument \citep[see e.g.][]{sachs1987}, which works when there is a position and momentum representation, but has certain shortcomings. I will then turn to what I take to be a better and much more general way to understand the origin of antiunitarity, which stems from the work of Wigner.

\subsection{The position-momentum approach to time reversal}\label{subs:pos-mom} Suppose we are dealing with a system involving position $Q$ and momentum $P$ satisfying the canonical commutation relations, $(QP - PQ) = i\hbar$. Suppose we can agree that time reversal preserves position while reversing momentum, $TQT^{-1}=Q$ and $TPT^{-1} = -P$.  Then, applying time reversal to both sides of the commutation relation we find,
\begin{align*}
    Ti\hbar T^{-1} & = T(QP - PQ)T^{-1} = (TQT^{-1})(TPT^{-1}) - (TPT^{-1})(TQT^{-1})\\
                  & = -(QP - PQ) = -i\hbar.
\end{align*}
Since $i\hbar$ is a constant, this outcome is not possible if $T$ is unitary, since all unitary operators are linear. So, since $T$ is not unitary, it can only be antiunitary, following the discussion of the previous section.

Why is it that time reversal preserves position and reverses momentum in quantum mechanics? It is often suggested that we must simply do what is already done in classical mechanics. But why do we do that? And, even presuming we have a good grip on time reversal in classical mechanics, why should time reversal behave this way in quantum mechanics too? 

Craig Callender has argued that it is because of Ehrenfest's theorem. This clever idea can be made precise as follows. Ehrenfest's theorem says that for any quantum state, the expectation values of quantum position $Q$ and momentum $P$ satisfy Hamilton's equations as they evolve unitarily over time. This means in particular that $q:=\Inn{\psi,Q\psi}$ and $p:=\Inn{\psi,P\psi}$ can be viewed as canonical position and momentum variables given a quantum state $\psi$. Now, assume classical time reversal preserves this canonical position and reverses the sign of momentum, $(q,p)\mapsto(q,-p)$. Assume also that quantum time reversal corresponds to a transformation $\psi\mapsto T\psi$ that respects the classical definition, in that it satisfies $(q,p)\mapsto(q,-p)$ when $q$ and $p$ are defined as above in terms of expectation values. These assumptions amount to the requirement that for all $\psi$,
\begin{align*}
  \Inn{T\psi,QT\psi} & = \Inn{\psi,Q\psi}\\
  \Inn{T\psi,PT\psi} & = -\Inn{\psi,P\psi}.
\end{align*}
This implies\footnote{We have $\Inn{\psi,T^*QT\psi} = \Inn{\psi,Q\psi}$ and $\Inn{\psi,T^*PT\psi} = -\Inn{\psi,P\psi}$ for arbitrary $\psi$, which implies that $T^*QT=Q$ and $T^*PT=-P$ \citep[see e.g.][Theorem I of Volume II, Chapter XV \S 2]{messiah1999}. A technical qualification is needed due to the fact that $Q$ and $P$ are unbounded: by `arbitrary $\psi$' we mean all $\psi$ in the domain of the densely-defined operator $Q$ (and similarly for $P$, respectively).} that $TQT^*=Q$ and $TPT^*=-P$. On this sort of thinking Callender concludes that, ``[s]witching the sign of the quantum momentum, therefore, is necessitated by the need for quantum mechanics to correspond to classical mechanics'' \citep[p.266]{callender2000time}.

Callender's suggestion certainly helps to clarify the relationship between classical and quantum time reversal. It can be applied when we restrict attention to quantum systems with a position and momentum representation. However, it does require us to understand time reversal in classical mechanics before knowing its meaning in quantum mechanics. That is perhaps unusual if one takes quantum theory to be the more fundamental or correct description of nature. And, although he does not mention it, Callender's argument also relies on a particular correspondence rule, that quantum time reversal is a transformation $\psi\mapsto T\psi$ that gives rise to classical time reversal on expectation values. Although this assumption is plausible, it is not automatic.

There are other shortcomings of the position-momentum approach. Some have complained that it is `basis-dependent' in the sense of requiring particular position and momentum operators to be chosen \citep{BiedenharnSudarshan}. A more difficult worry in this vein is that many quantum systems do not even admit such operators, in the sense that they do not admit a representation of the canonical commutation relations. This is often the case in relativistic quantum theory, where localised position operators are difficult if not impossible to define\footnote{This is a consequence of a class of no-go theorems established by \citet{hegerfeld1994a}, \citet{malament1996}, \citet{halvclift2002particles}, and others.}.

It would be nice to have a more general way to understand why time reversal is antiunitary, without mere appeal to convention, and without appeal to classical mechanics. In what follows, I will point out one such account. Let me begin with a discussion of invariance.

\subsection{The meaning of invariance} Many laws of nature are associated with a set of dynamical trajectories, which are typically solutions to some differential equation. These solutions represent the possible ways that states of the world can change over time. We say that such a law is \emph{invariant} under a transformation if and only if this set of dynamical trajectories is preserved by that transformation. In other words, invariance under a symmetry transformation means that if a given dynamical trajectory is possible according to the law, then so is the symmetry-transformed trajectory.

The same thinking applies in the language of quantum theory. Let the dynamical trajectories of a general quantum system be unitary, meaning that an initial quantum state $\psi$ evolves according to $\psi(t) = e^{-itH}\psi$ for each real number $t$, where $H$ is a fixed self-adjoint operator, the `Hamiltonian'\footnote{This is the `integral form' of Schr\"odinger's equation: taking the formal derivative of both sides and multiplying by $i$ yields, $i\tfrac{d}{dt}\psi(t) = -i^2He^{-itH}\psi = H\psi(t)$.}. Suppose a symmetry transformation takes each trajectory $\psi(t)$ to a new trajectory $\phi(t)$. If each transformed trajectory is also unitary, in that $\phi(t) = e^{-itH}\phi$ for all $t$, then we say that the quantum system is \emph{invariant} under the symmetry transformation.

Our concern in this paper will be with invariance under transformations that correspond to `reversing time'. There are many of them: one can reverse time and also translate in space; reverse time and also rotate; and so on. But these transformations share the property that, in addition to however they transform a state $\psi$ (possibly by the identity), they also reverse the order of states in each trajectory.

Call the latter `time-order reversal'. What exactly does that mean to reverse the time-order of a quantum trajectory $\psi(t)$? For example, $\psi(t)\mapsto\psi(-t)$ reverses time order, but so does $\psi(t)\mapsto\psi(1/e^t)$. Which is correct? A first guiding principle is that time-order reversal should not change the duration of time between any two moments; otherwise it would do more than just order reversal\footnote{No such criterion is adopted by \citet{Peterson2015a}, which leads him to consider a wealth of non-standard ways to reverse order in time.}. To enforce this we take time-order reversal to be a linear transformation of the reals, $t\mapsto at + b$ for some real $a,b$. A second guiding principle is to take `reversal' to mean that two applications of the transformation are equivalent to the identity transformation; this is to say that $t\mapsto at + b$ is an involution. The only order-reversing linear involutions of the reals have the form $t\mapsto -t + t_0$ for some real $t_0$. So, since the quantum theories we are concerned with here are time translation invariant, we may set $t_0=0$ without loss of generality, and take time-order reversal to have the form $\psi(t)\mapsto\psi(-t)$ as is usually presumed.

This time-order reversal must now be combined with a bijection $T:\HH\rightarrow\HH$ on instantaneous states. So, the time-reversing transformations can be minimally identified as bijections on the set of trajectories $\psi(t)=e^{-itH}\psi$ that take the form,
\[
  \psi(t) \mapsto T\psi(-t) = Te^{itH}\psi,
\] 
where $T$ at this point is an arbitrary unitary or antiunitary operator, possibly even the identity operator $I\psi:=\psi$. As with general symmetry transformations, we say that a quantum system is \emph{invariant} under these `$T$-reversal' transformations if and only if each trajectory $T\psi(-t)$ can be expressed as a trajectory $\phi(t)$ that satisfies the same unitary law, $\phi(t) = e^{-itH}\phi$. This statement can be summarised in a convenient form. Defining $\phi(t) := T\psi(-t) = Te^{itH}\psi$ (and hence that $\phi:=T\psi$), \emph{$T$-reversal invariance} means that,
\[
  Te^{itH}\psi = e^{-itH}T\psi
\]
for all $\psi$.

\subsection{A general approach to time reversal} Suppose that we know almost nothing about some $T$-reversal transformation, other than that it is unitary or antiunitary. But let us suppose that, whenever this transformation represents `the reversal of the direction of time', possibly together with other transformations too, then there is at least one realistic dynamical system that is $T$-reversal invariant in the sense defined above. Here is what that means in more precise terms. A realistic dynamical system requires a Hamiltonian that is not the zero operator, since otherwise no change would occur at all. Moreover, all known Hamiltonians describing realistic quantum systems are bounded from below, which we will express by choosing a lower bound of $0 \leq \Inn{\psi,H\psi}$. Finally, suppose that at least one of those Hamiltonians satisfies the $T$-invariance property that $Te^{itH}\psi = e^{-itH}T\psi$. Of course, some Hamiltonians will fail to satisfy this, such as those appearing in the theory of weak interactions, and this is perfectly compatible with our argument. However, we do suppose that \emph{at least one} of these Hamiltonians --- perhaps a particularly simple one with no interactions --- is $T$-reversal invariant. This turns out to be enough to establish that $T$ is antiunitary.\footnote{This proposition, a version of which is given by \citet[\S 2]{roberts-dissertation}, makes precise a strategy that was originally suggested by \citet[\S 20]{wigner1931}. I thank David Malament for suggestions that led to improvements in this formulation.}.

\begin{prop}\label{prop:1}
Let $T$ be a unitary or antiunitary bijection on a separable Hilbert space $\HH$. Suppose there exists at least one densely-defined self-adjoint operator $H$ on $\HH$ that satisfies the following conditions.
	\begin{enumerate}
		\item[(i)] (positive) $0 \leq \Inn{\psi,H\psi}$ for all $\psi$ in the domain of $H$.
		\item[(ii)] (non-trivial) $H$ is not the zero operator.
		\item[(iii)] ($T$-reversal invariant) $Te^{itH}\psi = e^{-itH}T\psi$ for all $\psi$.
	\end{enumerate}
Then $T$ is antiunitary.
\end{prop}
\begin{proof}
Condition (iii) implies that $e^{itH} = Te^{-itH}T^{-1} = e^{T(-itH)T^{-1}}$. Moreover, Stone's theorem guarantees the generator of the unitary group $e^{itH}$ is unique when $H$ is self-adjoint, so $itH = -TitHT^{-1}$. Now, suppose for reductio that $T$ is unitary, and hence linear. Then we can conclude from the above that $itH = -itTHT^{-1}$, and hence $THT^{-1}=-H$. Since unitary operators preserve inner products, this gives, $\Inn{\psi,H\psi} = \Inn{T\psi,TH\psi} = -\Inn{T\psi,HT\psi}$. But Condition (i) implies both $\Inn{\psi,H\psi}$ and $\Inn{T\psi,HT\psi}$ are non-negative, so we have,
\[
  0 \leq \Inn{\psi,H\psi} = -\Inn{T\psi,HT\psi} \leq 0.
\]
It follows that $\Inn{\psi,H\psi} = 0$ for all $\psi$ in the domain of $H$. Since $H$ is densely defined, this is only possible if $H$ is the zero operator, contradicting Condition (ii). Therefore, since $T$ is not unitary, it can only be antiunitary.
\end{proof}

This proposition applies equally to both non-relativistic quantum mechanics and to relativistic quantum field theory. It can also be straightforwardly extended to contexts in which energy is negative, by replacing Premise (i) with the (i${}^*$): the spectrum of $H$ is bounded from below but not from above; then the argument above proceeds in exactly the same way\footnote{In particular, we get $r\leq\Inn{\psi,H\psi}$ and $r\leq\Inn{T\psi,HT\psi}$, so $r \leq \Inn{\psi,H\psi} = -\Inn{T\psi,HT\psi} \leq -r$, which contradicts the assumption that the spectrum of $H$ is unbounded from above. I thank David Wallace for a discussion that led to this variation.}.

This line of argument is somewhat more abstract than the commutation relations approach. However, it is much more general. We do not need to appeal to time reversal in classical mechanics, or even have a representation of the commutation relations. We do not even presume that time reversal transforms instantaneous states by anything other than the identity. But we do \emph{derive} that it does, contrary to the second myth discussed at the outset. The derivation hinges on the presumption that there is at least one possible dynamical system --- not necessarily even one that \emph{actually} occurs! --- that is time reversal invariant. If there is, then time reversal can only be antiunitary, \emph{pace} the misgivings of the authors discussed above.

A similar but slightly stronger presumption has been advocated by \citet[\S 1.4]{sachs1987}, which he calls ``kinematic admissibility''. According to Sachs, ``[i]n order to express explicitly the independence between the kinematics and the nature of the forces, we require that the transformations leave the equations of motion invariant \emph{when all forces or interactions vanish}'' \citep[p.7]{sachs1987}. Requiring that admissible symmetry transformations have this property is equivalent to requiring that the Hamiltonian for free particles and fields is preserved by such symmetries. This is a special case of what we have assumed above, since free Hamiltonians are generally non-trivial and positive. It is also quite reasonable in my view. However, we simply do not need it to establish the result above. Time reversal is antiunitary as long as there is \emph{some} positive non-trivial Hamiltonian that is time reversal invariant. Whether that turns out to be the free Hamiltonian is beside the point.

Although this account of time reversal does not come for free, I think it does help to clarify what's at stake in debates like those of \citet{callender2000time} and \citet{albert2000}. Denying that time reversal is antiunitary means that a non-trivial realistic quantum system is never time reversal invariant, under any circumstances whatsoever. Even an empty system with no interactions would be asymmetric in time. Earman has called the disparity this creates with respect to the time symmetry\footnote{In fact the problem is more complicated: the Callender and Albert arguments seemingly entail that one must reject the standard antisymplectic time reversal operator in classical Hamiltonian mechanics as well, which would lead one to infer that classical Hamiltonian mechanics is not time reversal invariant even for the free particle. This avoids the perversity Earman identifies at the cost of introducing a new one: the failure of classical time reversal invariance!} found in classical mechanics ``the symptom of a perverse view'' \citep[p.249]{earman2002}. But even setting aside moral outrage, there is little practical use in identifying time reversal with a transformation that doesn't make any distinctions at all, not even between a free particle and one experiencing an important time-directed process like a weakly interacting meson. To those who value the alignment of philosophy of physics with the practice of physics, this may be too high a price to pay, especially for a view that is motivated by seemingly arbitrary metaphysics. 

\section{Third Stage: Transformation rules}

\subsection{Transformation rules} We now turn to what explains why time reversal preserves position, $Q\mapsto Q$, reverses momentum $P\mapsto -P$, and reverses spin $\sigma\mapsto-\sigma$. The fact that time reversal is antiunitary is not enough, since there are many such operators that do not do this. The commutation relations are not enough either\footnote{Example: if $[Q,P]\psi=i\psi$, then $[Q+P,P]\psi=i\psi$. But although both pairs $(Q,P)$ and $(Q+P,P)$ satisfy the canonical commutation relations, an antiunitary $T$ cannot preserve both $Q$ and $Q+P$ while also reversing $P$.}. So, what is the origin of these rules? The myth is that it can only be a matter of convention, or else an appeal to classical mechanics. This section will dissolve that final myth. There is a fairly general strategy for determining how time reversal will transform a given observable, which draws on how we understand the symmetries generated by that observable.

Let me start by uniquely deriving of how time reversal transforms position and momentum, then spin, and then discuss how the strategy can be applied to more general observables. Along the way, I will also give a new perspective on why $T^2=-1$ for quantum systems consisting of an odd number of fermions.

\subsection{Position and momentum transformations} 

Standard treatments of time reversal take the position and momentum transformation rules for granted. Non-standard treatments such as \citet{callender2000time} and \citet{albert2000} deny that these transformation rules hold in general\footnote{One might derive the $p\mapsto-p$ transformation rule on the non-standard view of time reversal whenever $dq/dt = p/m$ for some $m\neq0$. But this is not generally the case, for example in electromagnetism when velocity is a function of both momentum and electromagnetic potential.}, although Callender argues that the momentum transformation rule can be justified by appeal to a classical correspondence rule. In this section, I would like to point out that one can go beyond both of these treatments. The standard transformation rules can be derived from plausible assumptions about the nature of time, without appeal to classical mechanics (or any other theory). Our account makes this possible because we have adopted an independent argument for antiunitarity above. Thus we are free to use antiunitarity in the derivation of the position and momentum transformation rules. This is exactly the opposite of the standard textbook argument described in Section \ref{subs:pos-mom}.

Begin with momentum, defined as the generator of spatial translations. The spatial translations are given by a strongly continuous one-parameter unitary representation $U_a$, with the defining property\footnote{Translations can be equivalently defined by $U_aE_\Delta U_a^* := E_{\Delta - a}$, where $\Delta\mapsto E_\Delta$ is the projection-valued measure associated with the position operator $Q$.} that if $Q$ is the position operator, then $U_aQU_a^* = Q + aI$, for all $a$. At the level of  wavefunctions in the Schr\"odinger representation, this group has the property that $U_a\psi(x) = \psi(x-a)$. In other words, translations quite literally `shift' the position of a quantum system in space by a real number $a$. This group can be written $U_a = e^{iaP}$ by Stone's theorem, and the self-adjoint generator $P$ is what we mean by momentum.

The strategy I'd like to propose begins by asking how the meaning of time reversal changes when we move to a different location in space. Let us take as a principled assumption that it does not. After all, the concept of `reversing time' should not have anything to do with where we are located in space. This means that if we first time reverse a state and then translate it, the result is the same as when we first translate and then time reverse,
\[
  U_aT\psi = TU_a\psi.
\]
Since $U_a = e^{iaP}$, this `homogeneity' of time reversal has implications for the momentum operator $P$. Namely,
\[
  e^{iaP} = Te^{iaP}T^{-1} = e^{T(iaP)T^{-1}} = e^{-iaTPT^{-1}}. 
\] 
where the final equality follows from the antiunitarity of $T$. This implies $TPT^{-1} = -P$, since the generator of $U_a=e^{iaP}$ is unique by Stone's theorem. Thus we have our first transformation rule: $TPT^{-1} = -P$. We do not need to take this fact for granted after all. It is encoded in the homogeneity of time reversal in space.

\citet{callender2000time} and \citet{albert2000} have expressed skepticism about the presumption that time reversal should do anything at all at an instant. Viewing $P(t)$ in the Heisenberg picture, this is to express skepticism that time reversal truly transforms $P\mapsto-P$. Let me emphasize again that we have not \emph{presumed} any such principle here: rather, we have derived it from more basic principles. Namely, we began with an argument that $T$ is antiunitary, and then showed that $P\mapsto-P$ follows so long as the meaning of time reversal does not depend on one's location in space.

From this perspective the transformation rule for $Q$ is even more straightforward: if time reversal does not depend on location in space, then we can equally infer that $TQT^{-1} = Q$. Alternatively, we could follow a strategy similar to the one above, by the considering the group of \emph{Galilei boosts} defined by $V_b = e^{iaQ}$. Here it makes sense to view time reversal as reversing the direction of a boost, just as the change in position of a body over time changes sign when we watch a film in reverse. In particular, if we time reverse a system and then apply a boost in velocity, then this is the same as if we had boosted in the opposite spatial direction and then applied time reversal, $V_bT\psi = TV_{-b}\psi$. Following exactly the same reasoning above we then find that $TQT^{-1} = Q$. Thus, from either the homogeneity of space or the reversal of velocities under time reversal, we derive the transformation rule $Q\mapsto Q$ as well.

In the Schr\"odinger `wavefunction' representation in which $Q\psi(x)=x\psi(x)$ and $P\psi(x)=i\tfrac{d}{dx}\psi(x)$, we can define the operator $T$ that implements these transformations as follows. Let $T=K$ be the conjugation operator in this representation, which is to say the operator that transforms each wavefunction $\psi(x)$ to its complex conjugate, $K\psi(x)= \psi(x)^*$. It follows immediately from this definition that $TQT^{-1}=Q$ and $TPT^{-1} = -P$. And it is not just that we \emph{can} define $T$ in this way. Given an irreducible representation of the commutation relations in Weyl form\footnote{The commutation relations in Weyl form state that $e^{iaP}e^{ibQ} = e^{iab}e^{ibQ}e^{iaP}$. This implies the ordinary commutation relations $[Q,P]\psi=i\psi$ but is expressed in terms of bounded operators.}, this characterisation of the time reversal operator is unique up to a constant:

\begin{prop}[uniqueness of $T$]
  Let $(U_a=e^{iaQ},V_b=e^{ibP})$ be a strongly continuous irreducible unitary representation of the commutation relations in Weyl form, and let $K$ be the conjugation operator in the Schr\"odinger representation. If $T$ is antilinear and satisfies $TU_aT^{-1}=U_{a}$ and  $TV_bT^{-1}=V_{-b}$, then $T=cK$ for some complex unit $c$.
\end{prop}
\begin{proof}
  For any such antilinear involution $T$, the operator $TK$ is unitary (since it is the composition of two antiunitary operators), and commutes with both $U_a$ and $V_b$. Therefore it commutes with the entire representation. But the representation is irreducible, so by Schur's lemma $TK=c$ for some $c\in\CC$, which is a unit $c^*c=1$ because $TK$ is unitary. So, since $K^2=I$, we may multiply on the right by $K$ to get that $cK=TK^2=T$ as claimed.
\end{proof}

\subsection{Spin observables}

Angular momentum can be defined as a set of generators of spatial rotations in a rest frame, and spin observables form one such set. This is what is meant when it is said that spin is a `kind of angular momentum'. For example, as is well known, the Pauli observables $\{I,\sigma_1,\sigma_2,\sigma_3\}$ for a spin-$\tfrac{1}{2}$ particle give rise to a degenerate group of spatial rotations,
\[
  R^j_\theta = e^{i\theta\sigma_j}, \;\; j=1,2,3,
\]
in which there are two distinct elements ($R^j_\theta$ and $-R^j_\theta$) for each spatial rotation. This owes to the fact that the group generated by the Pauli matrices is isomorphic to $SU(2)$, the double covering group usual group of spatial rotations $SO(3)$. Nevertheless, each operator $R^j_\theta$ generated by a Pauli observable $\sigma_j$ can be unambiguously interpreted as representing a rotation in space.

We can now adopt our strategy from before, and ask: how does the meaning of time reversal change under spatial rotations? And here again the answer should be `not at all', insofar as the meaning of `reversal in time' does not have anything to do with orientation in space. This means in particular that if we rotate a system to a new orientation, apply time reversal, and then rotate the system back again, the result should be the same as if we had only applied time reversal in the original orientation. Equivalently, time reversal commutes with spatial rotations,
\[
  R_\theta T\psi = TR_\theta\psi.
\]
But from this it follows that $R_\theta = e^{i\theta\sigma_j} = Te^{i\theta\sigma_j}T^{-1} = e^{-i\theta T\sigma_j T^{-1}}$, which implies that $T\sigma_j T^{-1} = -\sigma_j$ for each $j=1,2,3$. As a result, time reversal transforms the Pauli spin observables as $\sigma_j \mapsto -\sigma_j$, as is standardly presumed. As with position and momentum, the spin transformation rules for time reversal are more than a convention: they arise directly out of the fact time reversal is isotropic in space. 

We can give the explicit definition of this $T$ by expressing the Pauli spin observables in the standard $z$-eigenvector basis as $\sigma_1=\paulix$, $\sigma_2=\pauliy$, $\sigma_3=\pauliz$. Let $K$ be the conjugation operator that leaves each of the basis vectors $\{\pzup,\pzdown\}$ invariant, but which maps a general vector $\psi=\alpha\pzup + \beta\pzdown$ to its conjugate $\psi^* = \alpha^*\pzup + \beta^*\pzdown$. Then one can easily check that the transformation $T = \sigma_2 K$ reverses the sign of each Pauli observable: since $K\sigma_2K^* = -\sigma_2$ and $K\sigma_iK^* = \sigma_i$ for $i=1,3$ we have,
\begin{align*}
  T\sigma_iT^* & = (\sigma_2K)\sigma_i(K\sigma_2) = \sigma_2(\sigma_i)\sigma_2 = -\sigma_i\\
  T\sigma_2T^* & = (\sigma_2K)\sigma_2(K\sigma_2) = \sigma_2(-\sigma_2)\sigma_2 = -\sigma_2.
\end{align*}
Thus, unlike time reversal in the Schr\"odinger representation, time reversal for spin is conjugation times an additional term $\sigma_2$ needed to reverse the sign of the matrices that don't have imaginary components. This immediately implies the famous result that, for a spin-$\tfrac{1}{2}$ particle, applying time reversal twice fails to bring you back to where you started, but rather results in a global change of phase:
\[
  T^2 = \sigma_2K\sigma_2K = \sigma_2(-\sigma_2)=-I.
\]
This property is in fact unavoidable: as before, there is a uniqueness result\footnote{A version of this was observed by \citet[Proposition 4]{roberts2012a}.} for the definition of $T$ for spin systems, which includes the fact that $T^2=-I$:

\begin{prop}[uniqueness of $T$ for spin]
  Let $\sigma_1$, $\sigma_2$, $\sigma_3$ be an irreducible unitary representation of the Pauli commutation relations, and let $K$ be the conjugation operator in the $\sigma_3$-basis. If $T$ is any antiunitary operator satisfying $Te^{i\sigma_j}T^{-1}=e^{i\sigma_j}$ for each $j=1,2,3$, then $T=c\sigma_2K$ for some complex unit $c$, and $T^2=-I$.	
\end{prop}
\begin{proof}
For any such antiunitary $T$, the operator $-T\sigma_2K$ is unitary (since it is the composition of two antiunitary operators) and commutes with all the generators $\sigma_1$, $\sigma_2$, $\sigma_3$. Thus it commutes with everything in the irreducible representation, so by Schur's lemma $-T\sigma_2K=cI$. This $c$ is a unit $c^*c=1$ because $-T\sigma_2K$ is unitary.  So, multiplying on the right by $\sigma_2K$ and recalling that $(\sigma_2K)^2=-I$ we get, $T = c\sigma_2K$, and hence $T^2 = (c\sigma_2 K)^2 = c^*c(\sigma_2K)^2 = -I$.
\end{proof}  

\subsection{Other observables} The examples above suggest a general strategy for determining how time reversal transforms an arbitrary observable in quantum theory. The strategy begins by considering the group of symmetries generated by an observable. We then ask how such symmetries change what it means to `reverse time'. The resulting commutation rule determines how time reversal transforms the original observable.

For example, in a gauge-invariant quantum system, it makes sense to presume that gauge transformations do not change the meaning of time reversal. In more precise terms this is to presume that the unitary gauge transformation $U = e^{i\Phi}$ commutes with the time reversal operator $T$. This implies the self-adjoint $\Phi$ that generates the gauge must reverse sign under time reversal, $T\Phi T^*=-\Phi$. The same observation holds whether we begin with the $U(1)$ gauge group of electromagnetism, or the $SU(3)$ gauge group of quantum chromodynamics. As soon as we have a grip on the way that time reversal transforms under a unitary symmetry group, we may immediately infer how it transforms the self-adjoint generators of that group.

\section{Conclusion}

Apart from dissolving some common mythology, I have advocated a perspective on time reversal according to which which its meaning is built up in three stages of commitment. The first stage commits to the direction of time being irrelevant to whether two states are mutually exclusive. This implies that time reversal is unitary or antiunitary. The second stage commits to at least one non-trivial, physically plausible system that is time reversal invariant. This guarantees that time reversal is antiunitary. The third stage commits to the meaning of time reversal being independent of certain symmetry transformations, such as translations or rotations in space. This gives rise to unique transformation rules for particular observables like position, momentum and spin.

Some may wish to get off the boat at any of these three stages. As with many things, the more one is willing to commit, the more one gets. However, I do not think even the strongest of these assumptions can be easily dismissed. The critics of the standard definition of time reversal have at best argued that time reversal should not transform instantaneous states. On the contrary, the perspective developed here shows a precise sense in which the non-standard perspective is implausible. Instantaneous properties of a physical system are sometimes temporally directed, and when this is the case, time reversal may transform them. As we have now seen, a few plausible assumptions about time in quantum theory give rise to just such a transformation, and indeed one that is in many circumstances unique. 

%\bibliographystyle{dcu} 
%\bibliography{/bwrtex/MasterBibliography}

\end{document}